\documentclass[10pt,journal,compsoc]{IEEEtran}
\usepackage{graphicx,subfigure}
\usepackage{amsmath}
\usepackage{amssymb,amsthm}
\ifCLASSOPTIONcompsoc
  \usepackage[nocompress]{cite}
\else
  \usepackage{cite}
\fi
\ifCLASSINFOpdf
\else
\fi
\begin{document}
\title{Network localization is unalterable by infections in bursts}

\author{Qiang~Liu and
        Piet~Van~Mieghem% <-this % stops a space
\IEEEcompsocitemizethanks{\IEEEcompsocthanksitem Q. Liu and P. Van Mieghem are with the Faculty of Electrical Engineering, Mathematics and Computer Science, Delft University of Technology, Delft, the Netherlands.\protect\\
E-mail: \{Q.L.Liu, P.F.A.VanMieghem\}@TUDelft.nl}
}

\IEEEtitleabstractindextext{%
\begin{abstract}
To shed light on the disease localization phenomenon, we study a bursty susceptible-infected-susceptible (SIS) model and analyze the model under the mean-field approximation. In the bursty SIS model, the infected nodes infect all their neighbors periodically, and the near-threshold steady-state prevalence is non-constant and maximized by a factor equal to the largest eigenvalue $\lambda_1$ of the adjacency matrix of the network. We show that the maximum near-threshold prevalence of the bursty SIS process on a localized network tends to zero even if $\lambda_1$ diverges in the thermodynamic limit, which indicates that the burst of infection cannot turn a localized spreading into a delocalized spreading. Our result is evaluated both on synthetic and real networks.
\end{abstract}

}

\maketitle

\IEEEdisplaynontitleabstractindextext

% \fi
\IEEEpeerreviewmaketitle

\IEEEraisesectionheading{\section{Introduction}\label{sec:introduction}}

\IEEEPARstart{T}{he} near-threshold behavior, i.e. the behavior around the threshold where a phase transition occurs, is of great interest in the study of dynamical processes, because many real complex systems may operate near the phase transition point \cite{nykter2008gene,furusawa2012adaptation,kitzbichler2009broadband}. One of the most extensively studied dynamical processes in network science is the susceptible-infected-susceptible (SIS) spreading process \cite{pastor2001epidemic,pastor2015epidemic}. For some networks, the SIS epidemic remains restricted into a small subnetwork and does not spread over the whole network for infection strength just above the (mean-field) epidemic threshold. This restricted spreading phenomenon is known as the (metastable) localization of the SIS process \cite{goltsev2012localization,ferreira2016metastable,de2017disease}, and has been studied recently. de Arruda \emph{et al.} \cite{de2017disease} investigated the localization phenomenon of SIS processes on multiplex networks. Sahneh \emph{et al.} \cite{sahneh2016delocalized} focused on the localization by a maximum entropy and optimization approach. Another near-threshold behavior, called \emph{Griffiths' phase} \footnote{The terminology \emph{Griffiths phase} is borrowed from the study of Ising ferromagnet. Griffiths finds that the magnetization of a random Ising ferromagnet is not an analytic function of external field $H$ at $H=0$ between the critical temperatures of the random and the corresponding pure Ising ferromagnet \cite{griffiths1969nonanalytic}, but in the study of epidemic processes, the non-analyticity of the function of the prevalence just above the epidemic threshold in the thermodynamic limit is still unknown.} of the SIS process, which is related to localization, is studied by Cota \emph{et al.} \cite{cota2016griffiths} and Mu{\~n}oz \emph{et al.} \cite{munoz2010griffiths}. The near-threshold behavior of the SIS process has also been applied to explain the operation of brain \cite{moretti2013griffiths}.

In this paper, we further study the SIS localization phenomenon. In previous studies \cite{goltsev2012localization,ferreira2016metastable}, localization of epidemic processes means that only a finite number of nodes is infected in the thermodynamic limit, i.e. when the network size $N\rightarrow \infty$. In this work, the definition of epidemic localization is that the average fraction of infected nodes, i.e. the prevalence, tends to zero in the thermodynamic limit, but the number of infected nodes is not necessarily finite. In the following part, we first clarify some misconceptions about the SIS localization in previous studies and show the availability of mean-field methods \cite{goltsev2012localization,ferreira2016metastable,van_mieghem_virus_2009,castellano2010thresholds}. We point out that the order of the near-threshold prevalence as a function of the network size $N$ is essential for understanding the influence of the network structure on spreading processes. Motivated by the essence of the prevalence order, we confine ourselves to a mean-field approximation and study a bursty spreading effect which maximizes the near-threshold prevalence by a factor equal to the largest eigenvalue $\lambda_1$ of the adjacency matrix of the network. Even though the spectral radius $\lambda_1$ diverges with network size $N$, the spreading bursts cannot change a localized spreading to a delocalized one if the principal eigenvector of the adjacency matrix of the network is localized.
\section{Misconceptions and conclusions about the epidemic localization}
In the SIS process, each node can be either infected or susceptible (healthy). An infected node can infect each healthy neighbor with infection rate $\beta$ and an infected node is spontaneously cured with curing rate $\delta$. The network is represented by the adjacency matrix $A$ with elements $a_{ij}$ for $i,j\in\{1,\ldots,N\}$. If node $i$ and $j$ are connected and $i\neq j$, then $a_{ij}=a_{ji}=1$; otherwise, $a_{ij}=a_{ji}=0$. The whole network can be in two different phases in the steady or metastable state: (a) in the all-healthy phase or (b) in the endemic phase. In the all-healthy phase, the epidemic has disappeared. In the endemic phase, the infection can persist in the network. The SIS process experiences a phase transition at a threshold \cite{van_mieghem_virus_2009,castellano2010thresholds}, which can be determined by the mean-field method $\tau_c^{(1)}=1/\lambda_1$. If the effective infection rate $\tau\triangleq \beta/\delta>\tau_c^{(1)}$, then the process is in the endemic phase under mean-field theory; otherwise, in the all-healthy phase.

For a finite network, the endemic and all-healthy phases can be identified by the \emph{prevalence}, which is the average fraction of infected nodes, and can be considered as an order parameter for the SIS process. A non-zero prevalence implies the endemic phase and a zero prevalence means the all-healthy phase. However, in the thermodynamic limit where the network size $N\rightarrow\infty$, a zero prevalence does not necessarily coincide with an all-healthy state just above the epidemic threshold. Goltsev \emph{et al.} \cite{goltsev2012localization} considered the zero prevalence in the thermodynamic limit as an indication of the localization phenomenon of the SIS process, where only a finite number of nodes are infected on average.
%Ferreira \emph{et al.} \cite{ferreira2016metastable} indicated that a finite number of infected nodes lead to the existence of the absorbing all-healthy state and thus there is only a metastable localization in the SIS process.
In particular, Goltsev \emph{et al.} \cite{goltsev2012localization} evaluate the steady-state prevalence $y_\infty(\tilde{\tau})$ just above the mean-field epidemic threshold by its first-order expansion $y_{\infty}(\tilde{\tau})=a\tilde{\tau}+o(\tilde{\tau})$ with \cite{van2012epidemic}
\begin{equation}\label{eq_SIS_critical_coefficient}
a=\frac{\sum_{i=1}^{N}x_i}{N\sum_{i=1}^{N}x_i^3}
\end{equation}
where $x_i$ is the $i$th component of the principal eigenvector of the adjacency matrix, obeying the normalized condition $\sum_{i=1}^{N}x_i^2=1$ and $\tilde{\tau}\triangleq \tau/\tau_c^{(1)}-1 \ll 1$ is the normalized effective infection rate. A tight bound of $a$ is $\frac{\min_i x_i}{\max_i x_i}<a<\frac{1}{\min_i x_i\sqrt{N}}$ as derived in Appendix~B. If $a\rightarrow 0$ as $N\rightarrow\infty$, then the near-threshold prevalence is zero, and if $a>0$ as $N\rightarrow\infty$, then a non-zero fraction of nodes are infected just above the threshold. Goltsev \emph{et al.} \cite{goltsev2012localization} define localization by the inverse participant ratio (IPR) $\eta(x)=\sum_{i=1}^{N}x_i^4$ of the principal eigenvector $x$, and state that if the IPR $\eta(x)=O(1)$, then the principal eigenvector $x$ is localized in a few components $x_i=O(1)$ and only a finite number of nodes are infected in the network with $a\rightarrow 0$ as $N\rightarrow\infty$. Otherwise, if $\eta(x)=o(1)$, then the vector $x$ is delocalized such that each component $x_i=O(\frac{1}{\sqrt{N}})$. Ferreira \emph{et al.} \cite{ferreira2016metastable} argue that if a finite number of nodes are infected using mean-field theory, then the virus eventually dies out and then the mean-field approximations \cite{pastor2001epidemic,van_mieghem_virus_2009} fail due to their omission of the absorbing state.

However, a zero prevalence in the thermodynamic limit does not necessarily mean that the number of infected nodes is finite. To illustrate this fact, let us consider a scale-free network which follows a power-law degree distribution with exponent $\gamma$, i.e. $\Pr[D=k]=\frac{k^{-\gamma}}{\zeta(\gamma)}$, $k\in\mathbb{N}$ and $\zeta(\gamma)$ is the Riemann zeta function \cite{titchmarsh1986theory}, in the thermodynamic limit. If the average degree of a scale-free network is finite, then $\gamma>2$ for $N\rightarrow\infty$, because $E[D^m]=\zeta(\gamma-m)/\zeta(\gamma)$ converges when $\gamma>m+1$. The maximum degree scales as $d_{\text{max}}=O(N^{1/(\gamma-1)})$ as derived in \cite[p.~594]{van_mieghem_performance_2014}, and thus we may find nodes with degree $O(N^\alpha)$ for $\alpha<1/(\gamma-1)$. Given a constant $c$, the expected number of nodes $\bar{n}_d$ with degree $d=[cN^\alpha]$ is $\bar{n}_d=N\Pr[D=[cN^\alpha]]=(c^{-\gamma}N^{1-\alpha\gamma})/\zeta(\gamma)$.
If $0<\alpha<\frac{1}{\gamma}$, then $\lim\limits_{N\rightarrow\infty}\bar{n}_d=\infty$. Thus, the average
number of hubs diverges. For each hub with degree of the order $O(N^\alpha)$ for $\alpha>0$, the local star subgraph ensures that the infection can persist for the effective infection rate $\tau>0$ in the thermodynamic limit. Related discussions can be found in \cite{chatterjee2009contact,boguna2013nature}, where the epidemic threshold of power-law networks is shown to be zero in the thermodynamic limit.

Furthermore, the principal eigenvector $x$ may not be localized in a finite subgraph, but localized in a subgraph whose size increases as $O(N^\alpha)$ with $0<\alpha<1$ with $N$. Pastor-Satorras and Castellano \cite{pastor2016distinct,pastor2018eigenvector} define the vector $x$ to be delocalized, only when the IPR $\eta(x)=O(N^{-1})$, while if $\eta(x)=O(N^{-\alpha})$ with $0\leq\alpha<1$, then $x$ is localized on a subgraph of size order of $O(N^\alpha)$. An example that can be exactly evaluated is the star-like, two-hierarchical graph \cite[p.~143]{van2010graph}. In this graph, there are $m$ fully connected nodes, and each node as hub is connected to $m$ leaf nodes. Basically, the graph consists of $m$ fully meshed $m$-stars. The network size is $N=m^2+m$ and the average degree is $d_{av}=3-\frac{4}{m+1}\approx 3$ for a large network. The largest eigenvalue $\lambda_1$ of the graph is $m$ as derived in \cite[p.~145]{van2010graph}, which is actually well approximated by the degree of each node in the maximum $K$-core \cite{castellano2017relating}. One may verify that the principal eigenvector
$$x=[\underbrace{\frac{1}{\sqrt{m+1}},\ldots,\frac{1}{\sqrt{m+1}}}_{m},\underbrace{\frac{1}{m\sqrt{m+1}},\ldots,\frac{1}{m\sqrt{m+1}}}_{m^2}]^T$$
is localized on a clique with size in the order of $O(1/\sqrt{N})$ and the IPR $\eta(x)=O(N^{-0.5})$. In this graph, the coefficient $a=O(\frac{1}{\sqrt{N}})$ leads to a zero prevalence, but the average number of infected nodes $Ny_\infty(\tilde{\tau})=O(\sqrt{N})$ diverges in the thermodynamic limit.

Even if the principal eigenvalue $x$ is localized in a finite subgraph and the IPR $\eta(x)=O(1)$, the average number of infected nodes may not be finite in the thermodynamic limit. Let us consider the extreme case of a star graph, whose principal eigenvector is $x=[\frac{1}{\sqrt{2}},\frac{1}{\sqrt{2(N-1)}}, \ldots, \frac{1}{\sqrt{2(N-1)}}]^T$. We may verify that the IPR $\eta(x)=O(1)$ and the coefficient $a=O(1/\sqrt{N})$. The average number of infected nodes is $Ny_\infty(\tilde{\tau})=O(\sqrt{N})$. Thus, just above the epidemic threshold (see also \cite{cator2013susceptible} for an exact, asymptotic analysis), an infinite number of nodes is infected, but the prevalence $y_\infty(\tilde{\tau})=O(\frac{1}{\sqrt{N}})$ tends to zero in the thermodynamic limit.

Our conclusions are: a) the localization of the principal eigenvector and the SIS epidemic process are related, but do not exactly correspond, because the infection can persist in subgraphs which correspond to the delocalized parts of the principal eigenvector; b) a zero prevalence just above threshold in the thermodynamic limit does not imply that the number of infected nodes is finite. Even for the star graph, the average number of infected nodes is of order $O(\sqrt{N})$ just above the epidemic threshold. Thus, it might be impossible to find a network, where the near-threshold number of infected nodes is finite in the thermodynamic limit under the mean-field theory. We address those conclusions to show that: a) in the thermodynamic limit, mean-field theories are consistent and applicable to study the near-threshold behavior because the epidemic may never die out; b) the order of the prevalence as a function of the network size $N$ is essential in the near-threshold spreading dynamic, which is also the motivation of our work. In the following part, we consider a network localized if the IPR $\eta(x)=O(N^{-\alpha})$ for $0\leq\alpha<1$, and is delocalized only if $\eta(x)=O(N^{-1})$ as defined by Pastor-Satorras and Castellano \cite{pastor2016distinct,pastor2018eigenvector}.

Throughout this paper, we confine ourselves to the mean-field method. Beyond the mean-field theory, the correlation between infection states of neighbors needs to be taken into consideration. In some cases, the correlation can be substantial. For example, the covariance of the infection state between neighbors in an infinite cycle graph is shown \cite[Theorem~3]{van2016approximate} to be $\xi=0.121375$ which is apparently not negligible and may introduce long-range correlations. The effect of long-range correlations on localization is unclear and the understanding of localization beyond mean-field theories is still open.
\section{The behavior of bursts just above the epidemic threshold}
Since our focus lies on the order of the prevalence as a function of network size $N$, we construct an SIS process with a non-constant prevalence in the steady state. We consider bursts that infect all healthy neighbors, leading to an explosion of the spreading. We choose periodical infections to allow analysis, and confine the SIS process to an infectious regime just above the epidemic threshold by tuning the period of the bursts. In some heterogeneous networks, e.g. scale-free networks, the ratio between the maximum prevalence (after each burst) and the minimum prevalence (before each burst) grows to infinity with the network size $N$. Even if infected nodes maximize their infection capability to infect all neighbors and magnify the prevalence by a divergent factor, we demonstrate that the process is still localized and the spreading is restricted to a small subgraph, whose size divided by the whole network size $N$ tends to zero.

In particular, our bursty SIS model is still an SIS model and each infected node can still be cured with rate $\delta$ as a Poisson process, but the infection (infecting all healthy neighbors) only happens at the time points: $1/\beta, 2/\beta, \ldots$ with infection rate $\beta$ and effective infection rate $\tau=\beta/\delta$. This bursty SIS model is a limit case of a non-Markovian SIS model \cite{liu2018burst} and was proposed to find the largest possible non-Markovian epidemic threshold. The bursty effect may lead to counterintuitive results. For example, in the epidemic process on a very large star graph, the infection probability of the hub node is much larger than those of the leaf node, when the process is just above the epidemic threshold. If the hub is infected just before a burst, the hub can infect all the leaf nodes and thus all nodes in the network are infected, which seems to lead to a non-zero prevalence (a global epidemic). However, even for the star graph, we will show that the prevalence just above threshold still converges to zero as the network size $N\rightarrow \infty$.

%Since the order of magnitude of the prevalence is of our interest, we consider bursts which allow the infected nodes to infect all their neighbors at certain time points boosting the prevalence into a higher order of magnitude, and then evaluate the order of the prevalence relative to network size $N$. We let the infection be periodical to make the model analyzable. In particular, our model is still an SIS model and each infected node can still be cured with rate $\delta$ as a Poisson process, but the infection (infecting all healthy neighbors) only happens at the time points: $1/\beta, 2/\beta, \ldots$ with infection rate $\beta$ and effective infection rate $\tau=\beta/\delta$. This bursty SIS model is a limit case of the non-Markovian SIS model \cite{liu2018burst} and was proposed to find the largest possible non-Markovian epidemic threshold. The bursty effect may lead to counterintuitive results. For example, in the epidemic process on a very large star graph, the infection probability of the hub node is much larger than those of the leaf node when the process is critical. Furthermore, if the hub is infected just before a burst, the hub can infect all the leaf nodes and thus all nodes in the network are infected, which seems to lead to a non-zero prevalence (a global epidemic). However, we will show that the critical prevalence still converges to zero as the network size $N\rightarrow \infty$.

The mean-field governing equations of the bursty SIS process are \cite{liu2018burst},

\begin{equation}\label{eq_burst_governing}
\begin{aligned}
v_i\left(\frac{n+1}{\beta}\right)&=\lim\limits_{t^*\rightarrow 1/\beta}\Bigg(\left[1-v_i\left(t^*+\frac{n}{\beta}\right)\right]\Bigg\{1-\\
&\prod_{j\in\mathcal{N}_i}\left[1-v_j\left(t^*+\frac{n}{\beta}\right)\right]\Bigg\}+v_i\left(t^*+\frac{n}{\beta}\right)\Bigg)
\end{aligned}
\end{equation}

and
\begin{equation}\label{eq_curing_governing}
\frac{\mathrm{d}v_i\left(\frac{n}{\beta}+t^*\right)}{\mathrm{d}t^*}=-\delta v_i\left(\frac{n}{\beta}+t^*\right)
\end{equation}
where $v_i(t)$ is the infection probability of node $i$ at time $t$, the length of the time passed after the nearest burst is $t^*\in [0,1/\beta)$, and $\mathcal{N}_i$ denotes the set of neighbors of node $i$. Equation~(\ref{eq_burst_governing}) and Eq.~(\ref{eq_curing_governing}) describe the bursty infection and curing processes, respectively. The epidemic threshold of the bursty SIS model is $\tau_c^{(B)}= 1/\ln(\lambda_1+1)$ as demonstrated in \cite{liu2018burst}. Figure~\ref{fig_prevalencevsT} shows that, if the effective infection rate $\tau$ is above the mean-field threshold $\tau_c^{(B)}$, then the prevalence periodically changes with period $1/\beta$ in the steady state; otherwise, the infection vanishes exponentially fast.
\begin{figure*}[t]
  \centering
  \subfigure[]{\label{fig_prevalencevsT}\includegraphics[width=0.4\textwidth]{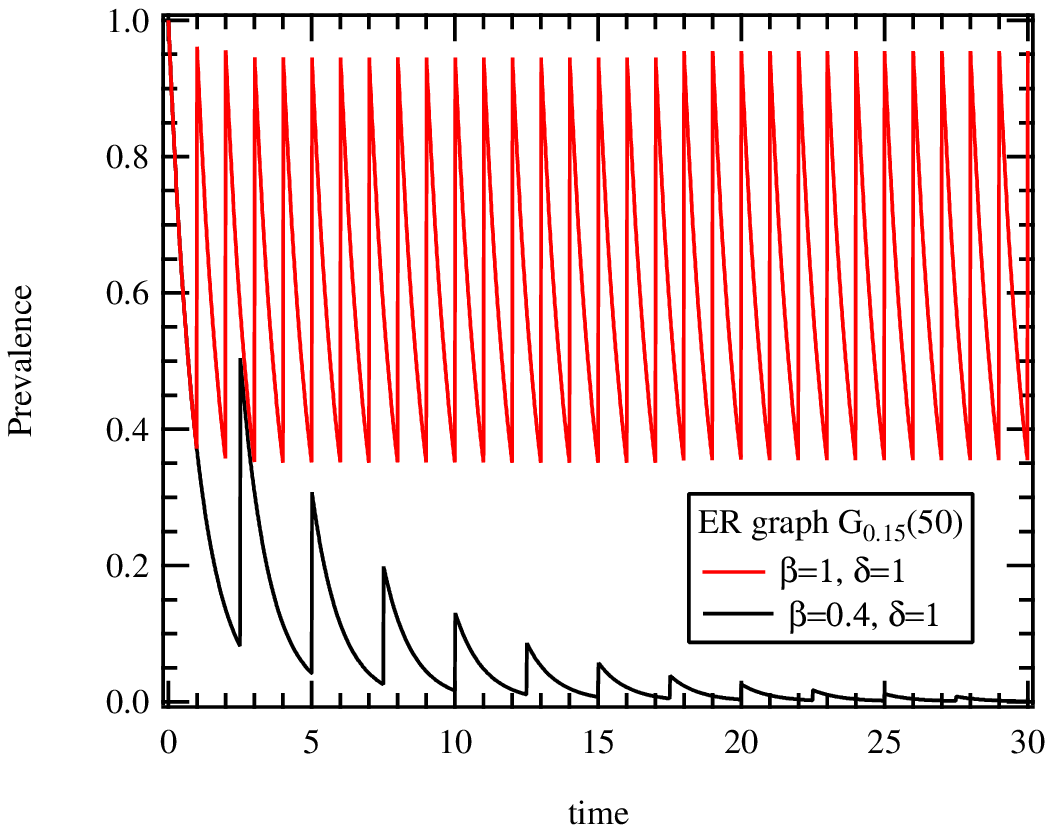}}
  \subfigure[]{\label{fig_phasetransition}\includegraphics[width=0.4\textwidth]{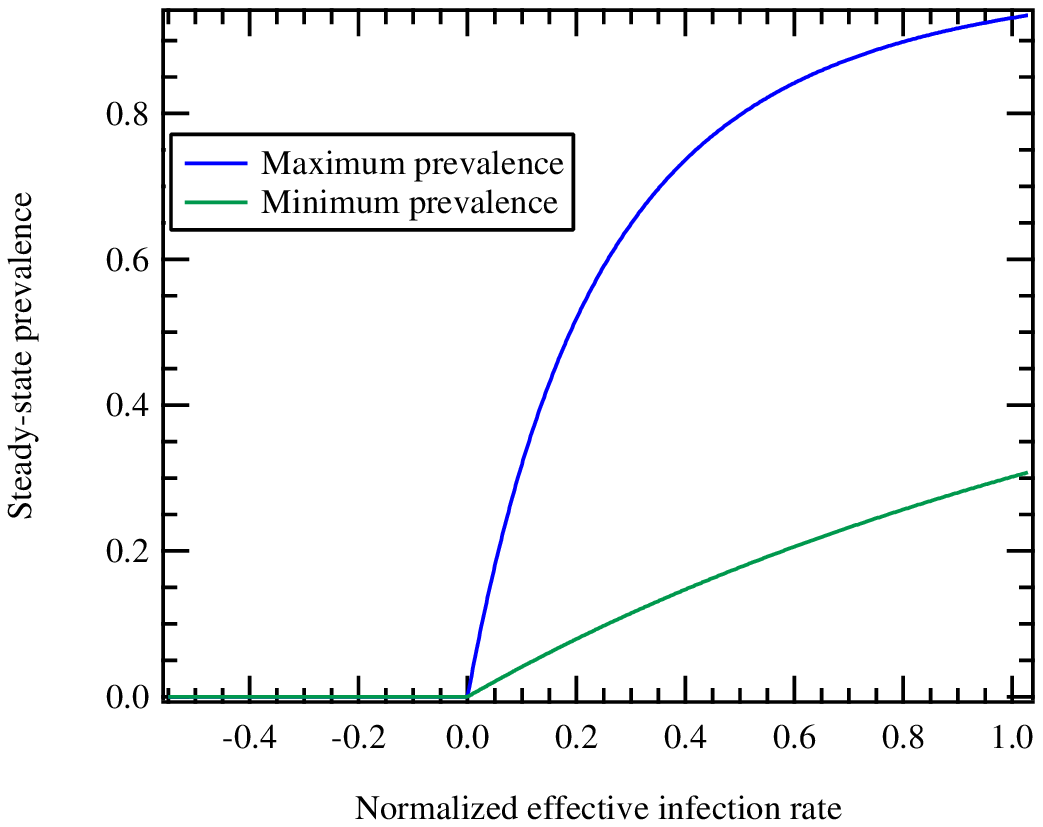}}
  \caption{(\textbf{a}): The bursty SIS prevalence on an Erd\H{o}s-R\'enyi (ER) graph $G_{0.15}(50)$. The epidemic threshold is $\tau_c^{(B)}=\frac{1}{\ln(\lambda_1+1)}=0.4437$. The red curve reflects the regime with the effective infection rate $\tau=1>\tau_c^{(B)}$, while the black curve represents the prevalence at $\tau=0.4$ below the threshold; (\textbf{b}): The phase transition of the bursty SIS model with the normalized effective infection rate $\tilde{\tau}$ on the same network. The upper blue curve and the lower green curve are the maximum and minimum steady-state prevalence, respectively. The steady-state prevalence changes periodically between the maximum and minimum.}
\end{figure*}
%\begin{figure}
% \includegraphics[width=0.4\textwidth]{prevalencevsT}%
% \caption{The prevalence of the bursty SIS process on an Erd\H{o}s-R\'enyi (ER) graph $G_{0.15}(50)$. The epidemic threshold is $\tau_c^{(1)}=0.4437$. The red curve reflects the regime with the effective infection rate $\tau=1>\tau_c^{(1)}$, while the black curve represents the prevalence at $\tau=0.4$ below the threshold. In the insert, we show the phase transition of the bursty SIS model with the normalized effective infection rate $\tilde{\tau}$ on the same network. The upper blue curve and the lower green curve are the maximum and minimum steady-state prevalence, respectively. The steady-state prevalence changes periodically between the maximum and minimum. \label{fig_prevalencevsT}}
%\end{figure}

We denote the steady-state prevalence at time $t^*$ after each burst by $y_\infty(\tilde{\tau},t^*)\triangleq \frac{1}{N}\lim\limits_{n\rightarrow\infty}\sum_{i=1}^{N}v_i(n/\beta+t^*)$ in the bursty SIS process with the normalized effective infection rate $\tilde{\tau}=\tau/\tau_c^{(B)}-1$. The steady-state prevalence $y_\infty(\tilde{\tau},t^*)$ is maximum just after each burst at $t^*=0$, denoted by $y_\infty^+(\tilde{\tau})\triangleq y_\infty(\tilde{\tau},0)$, and is minimum before each burst at $t^*\rightarrow 1/\beta$, denoted by $y_\infty^-(\tilde{\tau})\triangleq\lim\limits_{t^*\rightarrow1/\beta}y_\infty(\tilde{\tau},t^*)$. The ratio between the maximum and minimum steady-state prevalence is shown in \cite{liu2018burst} to be $y_\infty^+(\tilde{\tau})/y_\infty^-(\tilde{\tau})\leq \lambda_1+1$ and equality is achieved when $\tilde{\tau}\downarrow 0$. Thus, for a network with a largest eigenvalue $\lambda_1=O(N^\alpha)$ with $\alpha>0$, $y_\infty^+(\tilde{\tau})/y_\infty^-(\tilde{\tau})$ diverges for small $\tilde{\tau}$ in the thermodynamic limit, which is the most unusual feature of the bursty dynamic compared to traditional studies. As shown in Fig.~\ref{fig_phasetransition}, the steady prevalence $y_\infty^+(\tilde{\tau})$ (blue curve) and $y_\infty^-(\tilde{\tau})$ (green curve) experience a phase transition at the threshold $\tilde{\tau}=0$. Although the two curves approach each other from above to $\tau_c^{(B)}$, their ratio $y_\infty^+(\tilde{\tau})/y_\infty^-(\tilde{\tau})$ can diverge if $\lambda_1\rightarrow\infty$ in the thermodynamic limit.

The maximum and the minimum steady-state prevalence $y_\infty^+(\tilde{\tau})=a_{\text{max}}\tilde{\tau}+o\left(\tilde{\tau}\right)$ and $y_\infty^-(\tilde{\tau})=a_{\text{min}}\tilde{\tau}+o\left(\tilde{\tau}\right)$ just above threshold possess coefficients (see Theorem~1 in Appendix~A)
\begin{equation}\label{eq_max_critical_prevalence}
  a_{\text{max}}=\frac{2(\lambda_1+1)\ln(\lambda_1+1)}{\lambda_1}a
\end{equation}
and
$a_{\text{min}}=a_{\text{max}}/(\lambda_1+1)$, respectively. The coefficient $a$ of the traditional SIS prevalence in (\ref{eq_SIS_critical_coefficient}) is only determined by the first- and third-order moments of the principal eigenvector $x$ and the network size $N$, but the coefficients $a_{\text{max}}$ and $a_{\text{min}}$ are also related to the largest eigenvalue $\lambda_1$.

As mentioned, the bursts increase the prevalence by a factor of $\lambda_1$. For delocalized network with convergent maximum degree, we expect that the largest eigenvalue $\lambda_1=O(1)$ because $\lambda_1\leq \max\limits_{\forall \text{link}(i,j)}\sqrt{d_id_j}$ as shown in \cite[p.~48]{van2010graph}. Thus, the maximum and minimum prevalence are of the same order $O(1)$. There is always a non-zero average fraction of infected nodes just above the mean-field epidemic threshold in the thermodynamic limit.

Now we consider the localized networks. If the variance $\text{Var}[D]\rightarrow\infty$ as $N\rightarrow\infty$, then the largest eigenvalue $\lambda_1\geq\sqrt{\mathrm{Var}[D]+E^2[D]}$ diverges as shown in \cite[p.~47]{van2010graph}. Furthermore, a divergent maximum degree ensures the largest eigenvalue $\lambda_1\rightarrow\infty$ as $N\rightarrow\infty$, since $\lambda_1$ of the whole network is larger than that of the star subgraph with a divergent hub \cite[Eq.~(3.23)]{van2010graph}. In particular, the largest eigenvalue of a power-law network diverges in the thermodynamic limit \cite{chung2003eigenvalues}. The bursts magnify the traditional SIS coefficient $a$ in (\ref{eq_SIS_critical_coefficient}) by a divergent factor $\ln(\lambda_1+1)$ as shown by Eq.~(\ref{eq_max_critical_prevalence}), i.e. $a_{\text{max}}=2\ln(\lambda_1)a$. For the eigenvector localization as discussed in \cite{pastor2016distinct}, where the eigenvector $x$ is defined to be localized in a finite or infinite subnetwork, the coefficient $a$ in (\ref{eq_SIS_critical_coefficient}) follows an decay as $O(N^{-\epsilon})$ for $\epsilon>0$ and the maximum coefficient $a_{\text{max}}$ in (\ref{eq_max_critical_prevalence}) will also converge to zero as $a_{\text{max}}=O(N^{-\epsilon}\ln N)$ since $\ln\lambda_1<\ln N$. Although the bursts allow the infected nodes to infect all their healthy neighbors to reach as many nodes as possible in the network, the bursts cannot transform a zero prevalence to a non-zero prevalence in the thermodynamic limit.

\section{Numerical and simulation results}
In this section, we evaluate our conclusion in synthetic and real networks.

\subsection{Numerical results under the mean-field theory}
The first case is the delocalized networks. In regular graphs with average degree $d$, the largest eigenvalue $\lambda_1=d$ and the coefficients $a_{\text{max}}$ and $a_{\text{min}}$ are constant, only depending on degree $d$ as explained in Appendix~D. Figure~\ref{fig_ERnetcoeff} shows the results of the Erd\"{o}s-R\'enyi (ER) graphs with average degree $d_{\text{av}}=8$, and both the maximum and minimum coefficients $a_{\text{max}}$ and $a_{\text{min}}$ are in the order of $O(1)$ and independent of the network size $N$.

For localized networks with divergent largest eigenvalue $\lambda_1$, the ratio between the maximum and minimum prevalence $\lim\limits_{\tilde{\tau}\downarrow 0}y^+(\tilde{\tau})/y^-(\tilde{\tau})\rightarrow\infty$ in the thermodynamic limit. We first consider star graphs as already mentioned. We may verify (see Appendix~C) that the coefficients of star graphs follow $a_{\text{max}}=O(N^{-0.5}\ln N)$ and $a_{\text{min}}= O(N^{-1}\ln N)$. Although the average number of infected nodes both before and after each burst diverge, the maximum and minimum prevalence converges to zero as $N\rightarrow\infty$. We also generate connected scale-free networks with different power-law exponents $\gamma$ and average degree $d_{av}=8$ using the method introduced by Goh \emph{et al.} \cite{goh2001universal}. When generating the scale-free networks, we only preserve the largest connected component, because the original method of Goh \emph{et al.} does not guarantee a connected network. Figure~\ref{fig_powerlawnetcoeff} shows that the coefficient $a_{\text{max}}$ of power-law networks with different exponent $\gamma$ decays with the network size $N$. Furthermore, we consider networks with exponential degree distribution and use the network generating method in \cite{dorogovtsev2013evolution}. Initially, there are only $m$ nodes in the network, and each step a new node arrives. The new node is randomly connected to $m$ nodes of the current network (without preferential attachment as in the Barab\'asi-Albert model \cite{barabasi1999emergence}). The case $m=1$ introduced in \cite{dorogovtsev2013evolution} generates a uniform recursive tree \cite[16.2.2]{van_mieghem_performance_2014}. Following a same derivation as in \cite{dorogovtsev2013evolution}, the degree distribution of the network is $\Pr[D=k]=\frac{1}{1+m}(1+1/m)^{-k+m}$ for a network with average degree $d_{av}=2m$ in the thermodynamic limit. Figure~\ref{fig_EXPnetcoeff} shows the maximum coefficient $a_{\text{max}}$ of exponential networks with $m=1,2,4$, which decays with network size $N$.

For the synthetic networks, we can evaluate their near-threshold behavior by generating those networks with different size and check their order with the network size $N$. However, the size of a real network is fixed and the value of the coefficients $a_{\text{max}}$ and $a_{\text{min}}$ provide no information about the order of magnitude as a function of the network size $N$. To obtain insights from the value of $a_{\text{max}}$ in real networks, we generate random synthetic networks with a similar size, average degree, and degree distribution for each real network and compare the coefficients $a_{\text{max}}$ of the synthetic networks with those of the real networks. For most real networks, the degree distributions approximately follow a power law \footnote{Although there are debates that power-law networks are rare \cite{broido2018scale,stumpf2012critical}, the degree distribution of most real networks is linear in a log-log plot for several orders of magnitude, and then we can use synthetic power-law random graphs to approximate those real networks.} or exponential distribution. Thus, we can compare those real networks with the synthetic power-law and exponential networks mentioned above. Figure~\ref{fig_realnets} shows the value of the coefficient $a_{\text{max}}$ of real networks and corresponding synthetic networks, which are described in detail in the supplementary information. The value of the coefficients $a_{\text{max}}$ are similar in synthetic and real networks, especially for large networks. Thus, we conjecture that the near-threshold behavior of bursts is similar in real and synthetic networks.

\subsection{Simulations}
We emphasize that the exact coefficient $a_{\text{max}}$ is hard to obtain by simulations due to several reasons: a) The SIS process on finite-size networks has no sharp phase transition; b) Around the mean-field epidemic threshold, most realizations of the simulation die out (entering the absorbing all-healthy state) in a relatively short time. The time when the process is in the metastable state is hard to determine; c) The prevalence $y_{\infty}^+(\tilde{\tau})$ and the normalized effective infection rate $\tilde{\tau}=\tau/\tau_c^{(B)}-1$ are small just above the mean-field threshold, and the numerical error of the exact coefficient $y_{\infty}^+(\tilde{\tau})/\tilde{\tau}$ can be large (since $\tilde{\tau}\approx 0$). Thus, only an approximation of the coefficient $a_{\text{max}}$ can be obtained by simulations.

In our simulations of the bursty SIS process, all nodes are infected at time $t=0$ to prevent early die-out \cite{dieout}. If a node is infected at time $t$, then the infected node will be cured at time $t+T$ where $T$ is an exponential random variable with mean $1/\delta$ and all its neighbors will be infected at time $t+1/\beta$ if $T>1/\beta$. Each realization of the bursty SIS process runs for $50$ time units (simulations stop at $t=50$) which are long enough under our setting and $10^5$ realizations are simulated for each network. During the simulation of the bursty SIS process, the number of infected nodes is recorded every $0.01$ time unit for each realization and the prevalence is calculated by averaging all realizations. The coefficient $a_{\text{max}}$ is calculated by dividing the last local maximum of the recorded prevalence by $\tilde{\tau}$.

The simulation result on ER random graphs is shown in Fig.~\ref{fig_ERnetcoeff} for $\tilde{\tau}=0.0001$ and curing rate $\delta=4$. The results on power-law networks is shown in Fig.~\ref{fig_powerlawnetcoeff} for $\tilde{\tau}=0.1$ and $\delta=2$. We also perform the simulations on exponential networks as shown in Fig.~\ref{fig_EXPnetcoeff}, for $\tilde{\tau}=0.1$ with $\delta=1$ for $m=1,2$ and $\delta=2$ for $m=4$. The different settings of parameters $\tilde{\tau}$ and $\delta$ are based on the relaxation time of the process, i.e. the time that the prevalence curve approaches zero visually. In the cases of power-law and the exponential graphs, most of the realizations die out and the prevalence is calculated by averaging the realizations which do not die out at $t=45$. In the power-law and the exponential graphs, the simulation results are amazingly consistent with the mean-field theoretical results even though correlations of the infection state between neighbors are omitted in the mean-field analysis. In the ER graphs, the mean-field approximation does not perform well because the correlations play a role in sparse networks with homogeneous degree distribution \cite{liu2018autocorrelation}. However, the variation of the simulated coefficient $a_{\text{max}}$ with the network size $N$ agrees with the mean-field results: Fig.~\ref{fig_ERnetcoeff} indicates delocalization while Fig.~\ref{fig_powerlawnetcoeff} and Fig.~\ref{fig_EXPnetcoeff} indicate localization of the bursty SIS process.

\begin{figure*}[t]
  \centering
  \subfigure[]{\label{fig_ERnetcoeff}\includegraphics[width=0.4\textwidth]{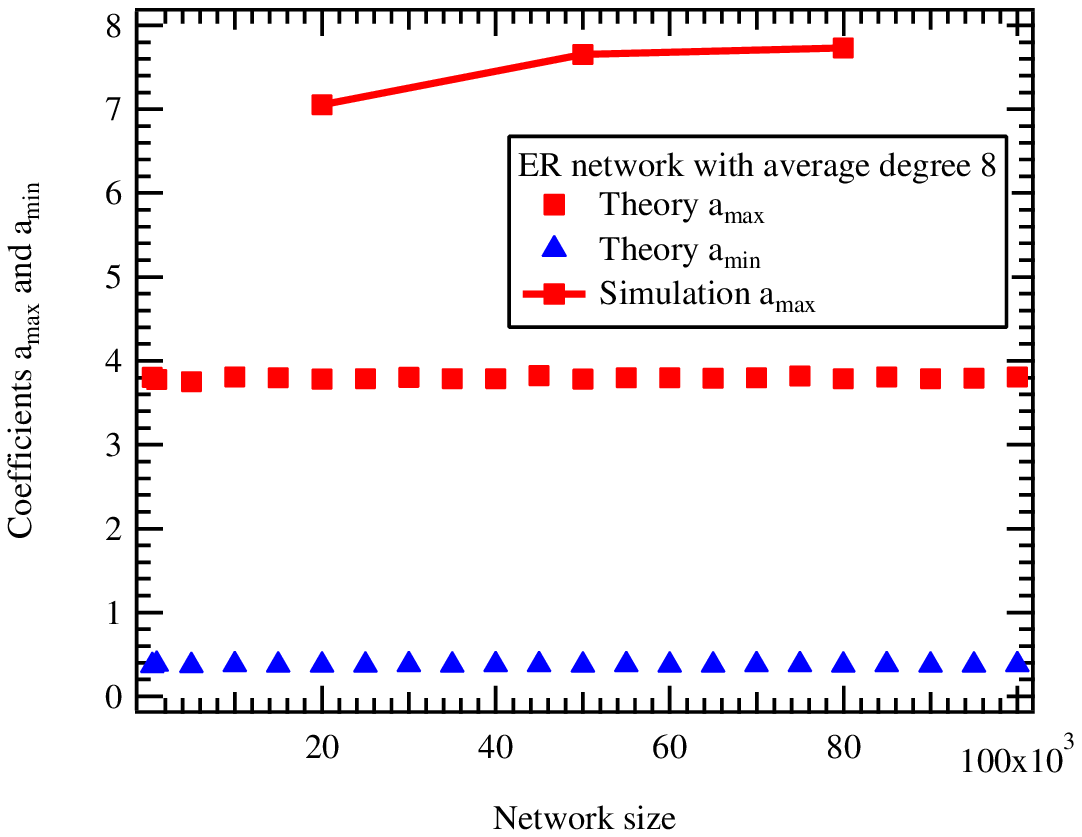}}
  \subfigure[]{\label{fig_powerlawnetcoeff}\includegraphics[width=0.4\textwidth]{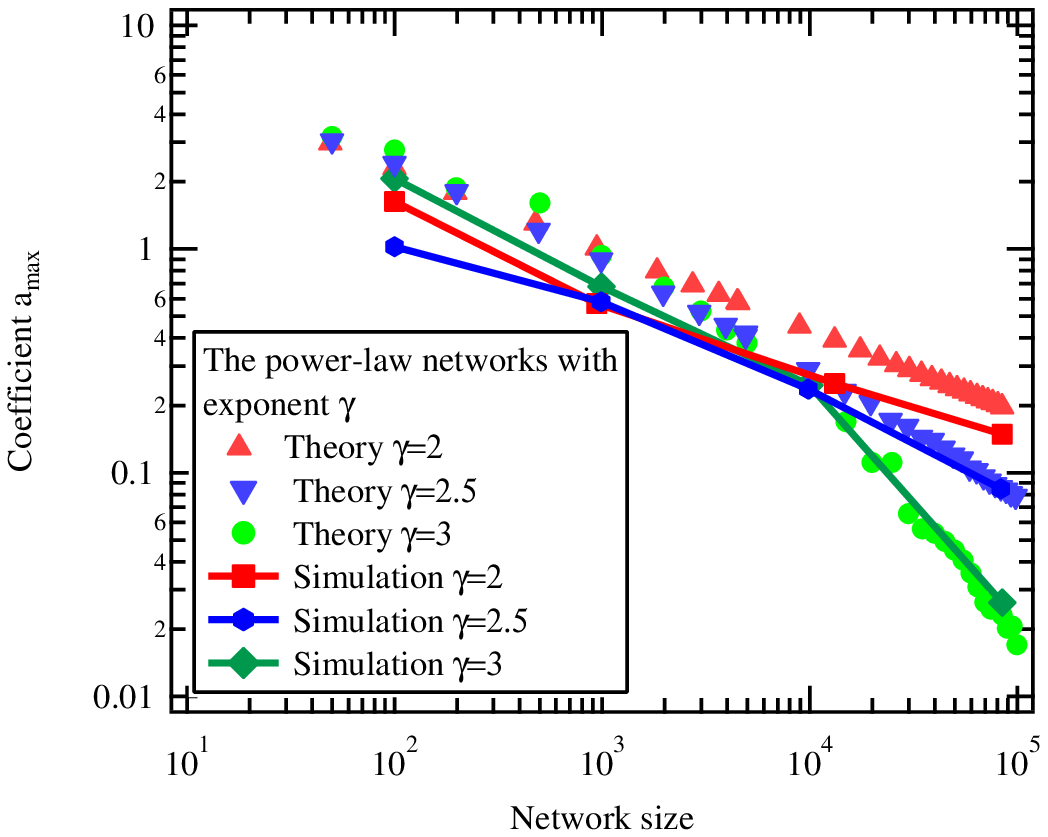}}\\
  \subfigure[]{\label{fig_EXPnetcoeff}\includegraphics[width=0.4\textwidth]{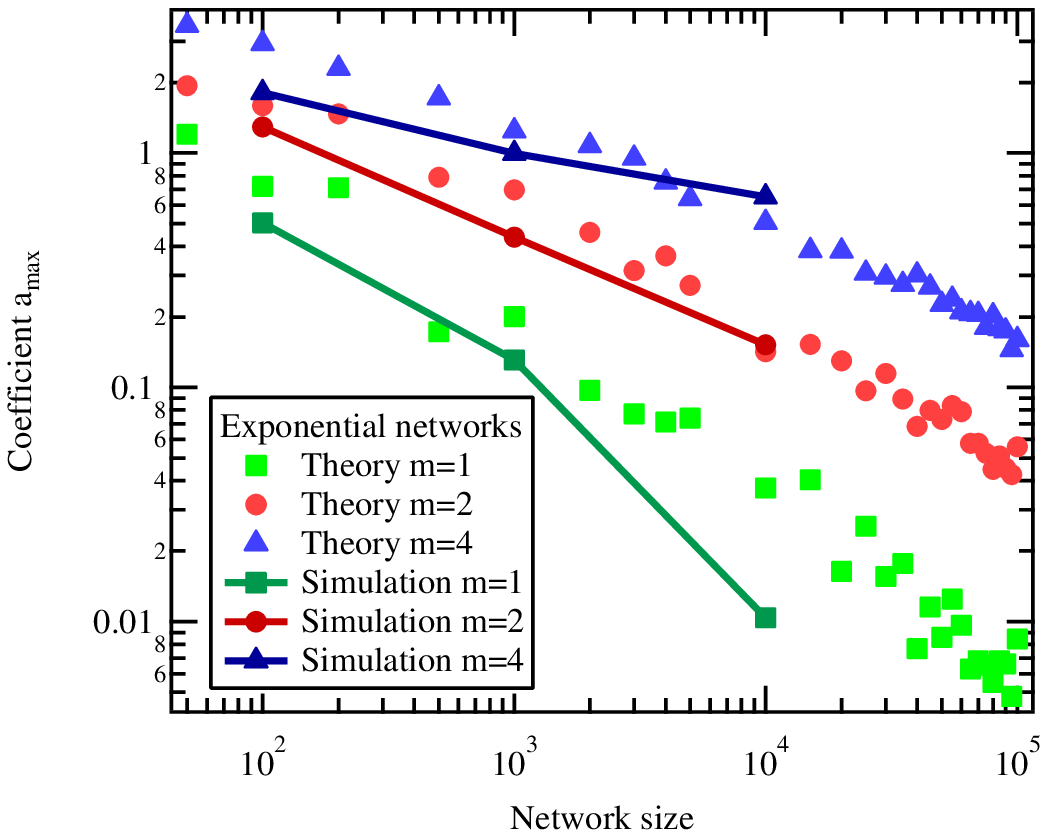}}
  \subfigure[]{\label{fig_realnets}\includegraphics[width=0.4\textwidth]{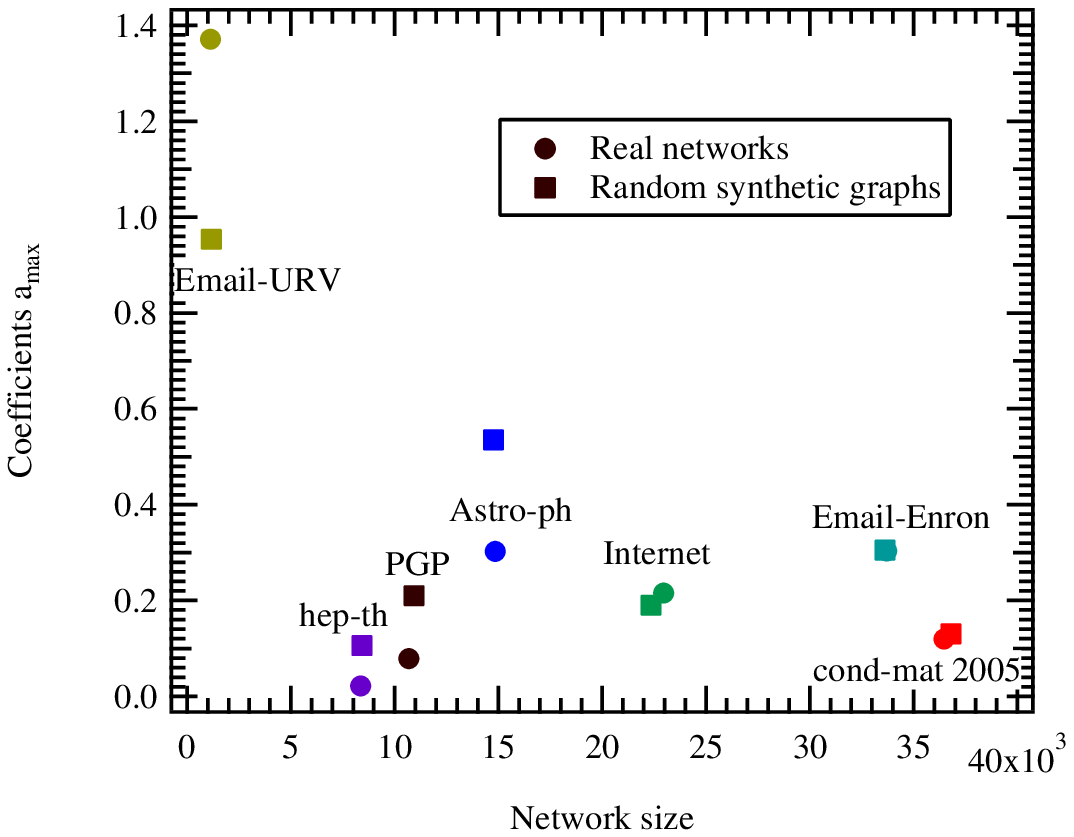}}
  \caption{(\textbf{a}): The coefficient $a_{\text{max}}$ and $a_{\text{min}}$ of ER networks; (\textbf{b}): The coefficients $a_{\text{max}}$ of networks with power-law degree distribution converge to zero with network size $N$; (\textbf{c}): The coefficient $a_{\text{max}}$ of networks with exponential degree distributions; (\textbf{d}): The coefficient $a_{\text{max}}$ of some well studied real networks: Email-URV \cite{guimera2003self}, hep-th \cite{newman2001structure}, PGP \cite{boguna2004models}, astro-ph \cite{newman2001structure}, Internet \cite{marknewmannet}, Email-Enron \cite{snapnets}, and cond-mat 2005 \cite{newman2001structure}.}
\end{figure*}
\section{Conclusion}
In this paper, we study the localization of the SIS process on networks. We specifically study a bursty SIS model which possesses a non-constant steady-state prevalence. In the bursty SIS model, the infected nodes can infect all healthy neighbors periodically to reach as many nodes as possible, and the prevalence is magnified by a divergent factor equal to the largest eigenvalue $\lambda_1$ in the thermodynamic limit. We show that the spreading process is still localized even if the bursty mechanism is applied, and our result introduces an open problem: are there any spreading dynamics leading to a delocalized spreading on networks with localized principal eigenvectors? If there exists such a case, then our analysis shows that the infection dynamic with a Poisson curing process must magnify the near-threshold prevalence $y_\infty(\tilde{\tau})$ of the traditional SIS model by a factor in the order of $O(N^z)$ for some value of $z\in(0,1)$.
\section*{acknowledgements}
 We thank Sergey Dorogovtsev and Karel Devriendt for their helpful discussion. Q.L. is thankful for the support from China Scholarship Council.
\ifCLASSOPTIONcaptionsoff
  \newpage
\fi
\bibliographystyle{IEEEtran}
\bibliography{bibitem}
\end{document}

\title{Appendix for "Network localization is unalterable by infections in bursts"}

\author{Qiang~Liu and
        Piet~Van~Mieghem% <-this % stops a space
\IEEEcompsocitemizethanks{\IEEEcompsocthanksitem Q. Liu and P. Van Mieghem are with the Faculty of Electrical Engineering, Mathematics and Computer Science, Delft University of Technology, Delft, the Netherlands.\protect\\
E-mail: \{Q.L.Liu, P.F.A.VanMieghem\}@TUDelft.nl}
}

\maketitle

\IEEEdisplaynontitleabstractindextext

% \fi
\IEEEpeerreviewmaketitle



\appendices
\section{The coefficient $a_{\text{max}}$}\label{append:coffa}
If the adjacency matrix of the network is $A$, the largest eigenvalue of $A$ is $\lambda_1$, the normalized principal eigenvalue of $A$ is $x=[x_1,\ldots,x_N]^T$, and the effective infection rate is $\tau=\beta/\delta$ with infection rate $\beta$ and curing rate $\delta$, then the epidemic threshold \cite[Theorem~1]{liu2018burst} of the bursty SIS model is $\tau_c^{(B)}=\frac{1}{\ln(\lambda_1+1)}$ and the following Theorem holds.
\begin{theorem}\label{theorem}
  For the bursty SIS process with effective infection rate $\tau$ above the threshold $\tilde{\tau}\triangleq\frac{\tau}{\tau_c^{(B)}}-1>0$, the maximum steady-state prevalence is $y^+_{\infty}(\tilde{\tau})=a_{\text{max}}\tilde{\tau}+o(\tilde{\tau})$ with
  $$a_{\text{max}}=\frac{2}{N}\frac{(\lambda_1+1)\ln(\lambda_1+1)\sum_{i=1}^{N}x_i}{\lambda_1\sum_{i=1}^{N}x_i^3}$$
  and the minimum prevalence is $y_{\infty}^-(\tilde{\tau})=a_{\text{min}}\tilde{\tau}+o(\tilde{\tau})$ with $a_{\text{min}}=a_{\text{max}}/(\lambda_1+1)$.
\end{theorem}

To prove Theorem~\ref{theorem}, we first prove the following Lemma.
\begin{lemma}\label{lemma}
  $$\sum_{i=1}^{N}x_i\sum_{\{j,k\in\mathcal{N}_i|j<k\}}x_jx_k+\lambda_1\sum_{i=1}^{N}x_i^3=\frac{1}{2}\lambda_1(\lambda_1+1)\sum_{i=1}^{N}x_i^3$$
  where $\mathcal{N}_i$ denotes the set of neighbors of node $i$.
\end{lemma}
\begin{proof}[Proof of Lemma~\ref{lemma}]
    For the first term on the left-hand side, we have
    \begin{eqnarray}
      \sum_{i=1}^Nx_i\sum_{\{j,k\in\mathcal{N}_i|j<k\}}x_jx_k&=&\frac{1}{2}\sum_{i=1}^{N}x_i\sum_{j\in\mathcal{N}_i}x_j\left(\sum_{k\in\mathcal{N}_i}x_k-x_j\right)\nonumber\\
       &=&\frac{1}{2}\sum_{i=1}^{N}x_i\sum_{j\in\mathcal{N}_i}x_j\sum_{k\in\mathcal{N}_i}x_k\nonumber\\
       &&-\frac{1}{2}\sum_{i=1}^{N}x_i\sum_{j\in\mathcal{N}_i}x_j^2\label{eq_temp22}
    \end{eqnarray}
    Since $\sum_{j\in\mathcal{N}_i}x_j=\lambda_1x_i$, the first term of (\ref{eq_temp22}) is $\frac{1}{2}\lambda_1^2\sum_{i=1}^{N}x_i^3$. We consider the second term of (\ref{eq_temp22})
    \begin{eqnarray}
      -\frac{1}{2}\sum_{i=1}^{N}x_i\sum_{j\in\mathcal{N}_i}x_j^2&=&-\frac{1}{2}\sum_{\forall \text{link} (i,j)}\left(x_i^2x_j+x_ix_j^2\right) \nonumber\\
      &=&-\frac{1}{2}\sum_{i=1}^{N}x_i^2\sum_{j\in\mathcal{N}_i}x_j\nonumber\\
      &=&-\frac{1}{2}\lambda_1\sum_{i=1}^{N}x_i^3\nonumber
      \end{eqnarray}
    Thus, the left-hand side equals $\frac{1}{2}\lambda_1(\lambda_1+1)\sum_{i=1}^{N}x_i^3$.
\end{proof}

\begin{proof}[Proof of Theorem~\ref{theorem}]
    The mean-field governing equations of the bursty SIS process are \cite{liu2018burst},
    \begin{equation}\label{eq_burst_governing}
    \begin{aligned}
    v_i\left(\frac{n+1}{\beta}\right)=&\lim\limits_{t^*\rightarrow 1/\beta}\Bigg(\left[1-v_i\left(t^*+\frac{n}{\beta}\right)\right]\Bigg\{1-\\
    &\prod_{j\in\mathcal{N}_i}\left[1-v_j\left(t^*+\frac{n}{\beta}\right)\right]\Bigg\}+v_i\left(t^*+\frac{n}{\beta}\right)\Bigg)
    \end{aligned}
    \end{equation}
    and
    \begin{equation}\label{eq_curing_governing}
    \frac{\mathrm{d}v_i\left(\frac{n}{\beta}+t^*\right)}{\mathrm{d}t^*}=-\delta v_i\left(\frac{n}{\beta}+t^*\right)
    \end{equation}
    where $v_i(t)$ is the infection probability of node $i$ at time $t$, $t^*\in [0,1/\beta)$ is the length of the time passed after the nearest burst, and $\mathcal{N}_i$ denotes the set of neighbor nodes of node $i$. The solution of Eq.~(\ref{eq_curing_governing}) is
    \begin{equation}\label{eq_solution_curing_process}
      v_i\left(\frac{n}{\beta}+t^*\right)=v_i\left(\frac{n}{\beta}\right)\mathrm{e}^{-\delta t^*}
    \end{equation}
    Substituting (\ref{eq_solution_curing_process}) at $t^*\rightarrow 1/\beta$, i.e. $\lim\limits_{t^*\rightarrow 1/\beta}v_i(n/\beta+t^*)=v_i(n/\beta)\exp(-1/\tau)$, into Eq.~(\ref{eq_burst_governing}), we obtain the following recursion of the infection probability of each node at $t^*=0$ just after each burst,
    \begin{equation}\label{eq_sync_SIS_max_infection_prob}
    \begin{aligned}
      v_i\left(\frac{n+1}{\beta}\right)=&\left(1-v_i\left(\frac{n}{\beta}\right)\mathrm{e}^{-1/\tau}\right)\Bigg(1-\\
      &\prod_{j\in\mathcal{N}_i}\left(1-v_i\left(\frac{n}{\beta}\right)\mathrm{e}^{-1/\tau}\right)\Bigg)+v_j\left(\frac{n}{\beta}\right)\mathrm{e}^{-1/\tau}
      \end{aligned}
    \end{equation}
    Equation~(\ref{eq_discrete_time_SIS_steady_state}) is the discrete-time SIS equation with infection probability $\tilde{\beta}=\mathrm{e}^{-1/\tau}$ and curing probability $\tilde{\delta}=1-\mathrm{e}^{-1/\tau}$. We rewrite Eq.~(\ref{eq_sync_SIS_max_infection_prob}) as,
    \begin{equation*}
    \begin{aligned}
      p_i[n+1]=&\left(1-(1-\tilde{\delta})p_i[n]\right)\left(1-\prod_{j\in\mathcal{N}_i}\left(1-\tilde{\beta}p_j[n]\right)\right)\\
      &+p_j[n](1-\tilde{\delta})
      \end{aligned}
    \end{equation*}
    where $p_i[n]\triangleq v_i(n/\beta)$. In the steady state, $\lim\limits_{n\rightarrow\infty}p_i[n+1]=\lim\limits_{n\rightarrow\infty}p_i[n]=p_{i\infty}$ for $1\leq i\leq N$, and we have,
    \begin{equation}\label{eq_discrete_time_SIS_steady_state}
      \tilde{\delta} p_{i\infty}=\left(1-(1-\tilde{\delta})p_{i\infty}\right)\left(1-\prod_{j\in\mathcal{N}_i}\left(1-\tilde{\beta} p_{j\infty}\right)\right)
    \end{equation}

    In the steady state, the discrete-time SIS infection probability vector $p_\infty\triangleq[p_{1\infty},\ldots,p_{N\infty}]$ approaches an eigenvector of the adjacency matrix $A$ corresponding to the largest eigenvalue $\lambda_1$ when $\tilde{\beta}/\tilde{\delta}\downarrow 1/\lambda_1$. Thus, we can assume $p_\infty=ax+o(a)q$, where $q$ is a vector orthogonal to $x$ and with finite components.
	
	Substituting $p_{\infty}=ax+o(a)q$ into (\ref{eq_discrete_time_SIS_steady_state}), we obtain,
    \begin{equation}\label{eq_temp5}
    \begin{aligned}
        \tilde{\delta} a x_i+\tilde{\delta} o(a)q_i=&\tilde{\beta} a\sum_{j\in\mathcal{N}_i}x_j+\tilde{\beta} o(a)\sum_{j\in\mathcal{N}_j}q_j-\\
        &a^2\tilde{\beta}^2\sum_{\{j,k\in\mathcal{N}_i|j<k\}}x_jx_k-\\
        &\tilde{\beta}(1-\tilde{\delta})a^2x_i\sum_{j\in\mathcal{N}_i}x_j+o(a^2)
        \end{aligned}
    \end{equation}
    where the eigenvalue equation indicates that $\sum_{j\in\mathcal{N}_i}x_j=\lambda_1x_i$.

    In vector form, (\ref{eq_temp5}) is,
    \begin{equation}\label{eq_temp6_1}
    \begin{aligned}
      \tilde{\delta} a x+\tilde{\delta} o(a) q=&\tilde{\beta} aAx+\tilde{\beta} o(a) Aq-\\
      &a^2\tilde{\beta}^2\text{vec}\left(\sum_{\{j,k\in\mathcal{N}_i|j<k\}}x_jx_k\right)-\\
      &\tilde{\beta}(1-\tilde{\delta})a^2\text{vec}\left(\lambda_1x_i^2\right)+o(a^2)h
      \end{aligned}
    \end{equation}
    where the vector $\text{vec}(z_i)\triangleq [z_1,\ldots,z_N]^T$.
    Divide both sides of (\ref{eq_temp6_1}) by $a\tilde{\beta}$ and recall that $Ax=\lambda_1 x$, and we have
    \begin{equation}\label{eq_temp6}
    \begin{aligned}
      \frac{\tilde{\delta}}{\tilde{\beta}} x+\frac{\tilde{\delta}}{\tilde{\beta}} \frac{o(a)}{a} q=& \lambda_1 x+\frac{o(a)}{a} Aq-\\
      &a\tilde{\beta}\text{vec}\left(\sum_{\{j,k\in\mathcal{N}_i|j<k\}}x_jx_k\right)-\\
      &a(1-\tilde{\delta})\text{vec}\left(\lambda_1x_i^2\right)+\frac{o(a^2)}{a}h
      \end{aligned}
    \end{equation}
    Rearranging (\ref{eq_temp6}), we obtain
    \begin{equation}\label{eq_temp7}
    \begin{aligned}
      \left(\lambda_1-\frac{\tilde{\delta}}{\tilde{\beta}}\right) x-\frac{\tilde{\delta}}{\tilde{\beta}} \frac{o(a)}{a} q= & -\frac{o(a)}{a} Aq+\\
      &a\Bigg[\tilde{\beta}\text{vec}\left(\sum_{\{j,k\in\mathcal{N}_i|j<k\}}x_jx_k\right)+\\
      &(1-\tilde{\delta})\text{vec}\left(\lambda_1x_i^2\right)\Bigg]+\frac{o(a^2)}{a}h
      \end{aligned}
    \end{equation}
    Since $a\rightarrow0$ as $(\lambda_1-\tilde{\delta}/\tilde{\beta})\rightarrow0$, we assume
     \begin{equation}\label{eq_a_and_a1_relation}
     	a=a_1(\lambda_1-\tilde{\delta}/\tilde{\beta})+o(\lambda_1-\tilde{\delta}/\tilde{\beta})
     \end{equation}
    and substitute $a$ into (\ref{eq_temp7}),
    \begin{equation}\label{eq_temp8}
    \begin{aligned}
      \left(\lambda_1-\frac{\tilde{\delta}}{\tilde{\beta}}\right) x-\frac{\tilde{\delta}}{\tilde{\beta}} \frac{o(a)}{a} q= & -\frac{o(a)}{a} Aq+\\
      &a_1\left(\lambda_1-\frac{\tilde{\delta}}{\tilde{\beta}}\right)d(\tilde{\beta},\tilde{\delta})+\\
      &o(a)d(\tilde{\beta},\tilde{\delta})+\frac{o(a^2)}{a}h
      \end{aligned}
    \end{equation}
    where $$d(\tilde{\beta},\tilde{\delta})=\left(\tilde{\beta}\text{vec}\left(\sum_{\{j,k\in\mathcal{N}_i|j<k\}}x_jx_k\right)+(1-\tilde{\delta})\text{vec}\left(\lambda_1x_i^2\right)\right).$$ We divide both side by $\lambda_1-\tilde{\delta}/\tilde{\beta}$,
    \begin{equation}\label{eq_temp9}
    \begin{aligned}
       x-\frac{\tilde{\delta}}{\tilde{\beta}} \frac{o(a)}{a}\frac{1}{\left(\lambda_1-\frac{\tilde{\delta}}{\tilde{\beta}}\right)} q= & -\frac{o(a)}{a}\frac{1}{\left(\lambda_1-\frac{\tilde{\delta}}{\tilde{\beta}}\right)} Aq+\\
       &a_1d(\tilde{\beta},\tilde{\delta})+\frac{o(a)}{\left(\lambda_1-\frac{\tilde{\delta}}{\tilde{\beta}}\right)}d(\tilde{\beta},\tilde{\delta})+\\
       &\frac{o(a^2)}{a\left(\lambda_1-\frac{\tilde{\delta}}{\tilde{\beta}}\right)}h
       \end{aligned}
    \end{equation}
    By taking the scalar product with $x$ on both sides of Eq.~(\ref{eq_temp9}) and recalling that the vector $q$ is orthogonal to the eigenvector $x$, we obtain
    \begin{equation}\label{eq_temp10}
       1=a_1d(\tilde{\beta},\tilde{\delta})\cdot x+\frac{o(a)}{\left(\lambda_1-\frac{\tilde{\delta}}{\tilde{\beta}}\right)}d(\tilde{\beta},\tilde{\delta})\cdot x+\frac{o(a^2)}{a\left(\lambda_1-\frac{\tilde{\delta}}{\tilde{\beta}}\right)}h\cdot x
    \end{equation}
    When $a\rightarrow0$, Eq.~(\ref{eq_temp10}) becomes
    \begin{equation}\label{eq_a_toward_0_eq}
    	1=a_1d\left(\frac{1}{\lambda_1}\tilde{\delta},\tilde{\delta}\right)\cdot x
    \end{equation}

    In the bursty SIS case where $\lim\limits_{\tau\downarrow\tau_c^{(B)}}\tilde{\delta}=\frac{\lambda_1}{\lambda_1+1}$, Eq.~(\ref{eq_a_toward_0_eq}) reads
    \begin{equation*}\label{eq_temp16}
       1=a_1d\left(\frac{1}{\lambda_1+1},\frac{\lambda_1}{\lambda_1+1}\right)\cdot x
    \end{equation*}
    Thus,
    \begin{equation*}
      a_1=\frac{\lambda_1+1}{\sum_{i=1}^Nx_i\sum_{\{j,k\in\mathcal{N}_i|j<k\}}x_jx_k+\lambda_1\sum_{i=1}^Nx_i^3}
    \end{equation*}
    Using Lemma~\ref{lemma}, $a_1$ becomes
    \begin{equation*}
      a_1=\frac{2}{\lambda_1\sum_{i=1}^{N}x_i^3}
    \end{equation*}
    We assume $a=a_2\epsilon+o(\epsilon)$ where $\epsilon=\tau-\tau_c^{(B)}=\tau-\frac{1}{\ln(\lambda_1+1)}$ and we may verify that
    $$\begin{aligned}\left.\frac{\mathrm{d}(\lambda_1-\tilde{\delta}/\tilde{\beta})}{\mathrm{d}\epsilon}\right|_{\epsilon=0}&=\left.\frac{\mathrm{d}(\lambda_1+1-\mathrm{e}^{1/\tau})}{\mathrm{d}\epsilon}\right|_{\epsilon=0}\\&=(\lambda_1+1)\ln^2(\lambda_1+1)\end{aligned}$$
    then we obtain
    \begin{equation*}\label{eq_temp18}
      a_2=a_1(\lambda_1+1)\ln^2(\lambda_1+1)=\frac{2(\lambda_1+1)\ln^2(\lambda_1+1)}{\lambda_1\sum_{i=1}^{N}x_i^3}
    \end{equation*}
    Thus, the maximum prevalence is $\frac{a_2\sum_{i=1}^{N}x_i}{N}\left(\tau-\tau_c^{(B)}\right)+o(\tau-\tau_c^{(B)})$.

    After normalizing the effective infection rate by $\tau/\tau_c^{(B)}$ and defining $\tilde{\tau}=\tau/\tau_c^{(B)}-1$, we finally find the maximum prevalence as
    \begin{eqnarray}\label{eq_temp23}
      y_{\infty}^+(\tilde{\tau})&=&\frac{a_2\tau_c^{(B)}\sum_{i=1}^{N}x_i}{N}\tilde{\tau}+o\left(\tilde{\tau}\right)\nonumber\\
      &=&\frac{2(\lambda_1+1)\ln(\lambda_1+1)\sum_{i=1}^{N}x_i}{N\lambda_1\sum_{i=1}^{N}x_i^3}\tilde{\tau}+o\left(\tilde{\tau}\right)
    \end{eqnarray}
    For general $t^*$, the prevalence is $\exp(-\delta t^*)y_{\infty}^+(\tilde{\tau})$ and then the minimum prevalence is $y_{\infty}^-(\tilde{\tau})=y_{\infty}^+(\tilde{\tau})/(\lambda_1+1)$ as $t^*\rightarrow 1/\beta$.
\end{proof}
\section{The bounds of $a$}\label{append:boundofa}
By the Perron-Frobenius theorem, every component of the principal eigenvector is positive. The lower bound of $a$ is derived follows.
\begin{equation*}
  a=\frac{\sum_{i=1}^{N}x_i}{N\sum_{i=1}^{N}x_i^3}\geq \frac{N\min\limits_i x_i}{N\max\limits_i x_i\sum_{j=1}^{N}x_j^2}=\frac{\min\limits_i x_i}{\max\limits_i x_i}
\end{equation*}
For the upper bound, using the Cauchy–Schwarz inequality $(\sum_{i=1}^{N}x_i)^2\leq N\sum_{i=1}^{N}=N$, we obtain
\begin{equation*}
  a=\frac{\sum_{i=1}^{N}x_i}{N\sum_{i=1}^{N}x_i^3}\leq \frac{\sqrt{N}}{N\min\limits_ix_i}=\frac{1}{\sqrt{N}\min\limits_ix_i}
\end{equation*}
The bound is tight when the network is a regular graph.
\section{The coefficients of star graphs}\label{append:star}
We may verify that the largest eigenvalue of the star graph is $\sqrt{N-1}$ and the principle eigenvector is $x=[\frac{1}{\sqrt{2}},\ldots,\frac{1}{\sqrt{2(N-1)}}]^T$. We have following results
\begin{eqnarray*}
  a=\frac{\sum_{i=1}^{N}x_i}{N\sum_{i=1}^{N}x_i^3} &=&\frac{1}{\sqrt{N}}+o(\frac{1}{\sqrt{N}}) \\
  a_{\text{max}}&=&\frac{2}{N}\frac{(\lambda_1+1)\ln(\lambda_1+1)\sum_{i=1}^{N}x_i}{\lambda_1\sum_{i=1}^{N}x_i^3} \\&=&\frac{\ln(\sqrt{N})}{\sqrt{N}}+o(N^{-\frac{1}{2}}\ln N) \\
  a_{\text{min}}&=&\frac{2}{N}\frac{\ln(\lambda_1+1)\sum_{i=1}^{N}x_i}{\lambda_1\sum_{i=1}^{N}x_i^3} \\&=&\frac{\ln(\sqrt{N})}{N}+o(N^{-1}\ln N)
\end{eqnarray*}
\section{The coefficients of $d$-regular graphs}\label{append:regular}
For regular graph, the principal eigenvector is $x=\frac{1}{\sqrt{N}}u$ where $u$ is all-one vector and the largest eigenvalue is $d$. We may verify that
\begin{eqnarray*}
  a=\frac{\sum_{i=1}^{N}x_i}{N\sum_{i=1}^{N}x_i^3} &=& 1 \\
  a_{\text{max}}=\frac{2}{N}\frac{(\lambda_1+1)\ln(\lambda_1+1)\sum_{i=1}^{N}x_i}{\lambda_1\sum_{i=1}^{N}x_i^3} &=& 2\left(1+\frac{1}{d}\right)\ln(d+1) \\
  a_{\text{min}}=\frac{2}{N}\frac{\ln(\lambda_1+1)\sum_{i=1}^{N}x_i}{\lambda_1\sum_{i=1}^{N}x_i^3} &=& \frac{2\ln(d+1)}{d}
\end{eqnarray*}
\section{Real networks}
The parameters of the real and synthetic networks are listed in Table~\ref{tab_real_networks}. The degree distributions are plotted in Fig.~\ref{cond-mat2005} to Fig.~\ref{EmailEnron}.

\begin{table*}
  \centering
  \begin{tabular}{|p{3cm}|p{5cm}|p{5cm}|}
  \hline
    Networks & Parameters of real networks & Parameters of corresponding synthetic networks\\
    \hline
    \hline
    Cond-mat 2005 \cite{newman2001structure}& $N=36458$, $d_{av}=9.4210$, $a_{\text{max}}=0.1199$& $N=36811$, $d_{av}= 9.4483$, $a_{\text{max}}=0.1301$\\
    \hline
    astro-ph \cite{newman2001structure}& $N=14845$, $d_{av}=16.1202$, $a_{\text{max}}=0.3024$& $N=14766$, $d_{av}= 16.2536$, $a_{\text{max}}=0.5352$\\
    \hline
    Internet \cite{marknewmannet} & $N=22963$, $d_{av}=4.2186$, $a_{\text{max}}=0.2155$& $N=22354$, $d_{av}= 4.2804$, $a_{\text{max}}=0.1903$\\
    \hline
    hep-th \cite{newman2001structure}& $N=5835$, $d_{av}=4.7352$, $a_{\text{max}}=0.0218$& $N=5944$, $d_{av}= 4.5855$, $a_{\text{max}}=0.1063$\\
    \hline
    Email-URV \cite{guimera2003self}& $N=1133$, $d_{av}=9.6222$, $a_{\text{max}}=1.3713$& $N=1178$, $d_{av}= 9.6774$, $a_{\text{max}}=0.9539$\\
    \hline
    PGP \cite{boguna2004models}& $N=10680$, $d_{av}=4.5536$, $a_{\text{max}}=0.0789$& $N=10986$, $d_{av}= 4.4773$, $a_{\text{max}}=0.2104$ \\
    \hline
    Email-Enron \cite{snapnets}& $N=33696$, $d_{av}=10.7319$, $a_{\text{max}}=0.3037$& $N=33632$, $d_{av}= 10.8451$, $a_{\text{max}}=0.3053$\\
    \hline
  \end{tabular}
  \caption{The parameters of real networks and the corresponding synthetic networks. Only the largest connected components are preserved and all the networks are connected.}\label{tab_real_networks}
\end{table*}

\begin{figure}
  \centering
  \includegraphics[scale=0.5]{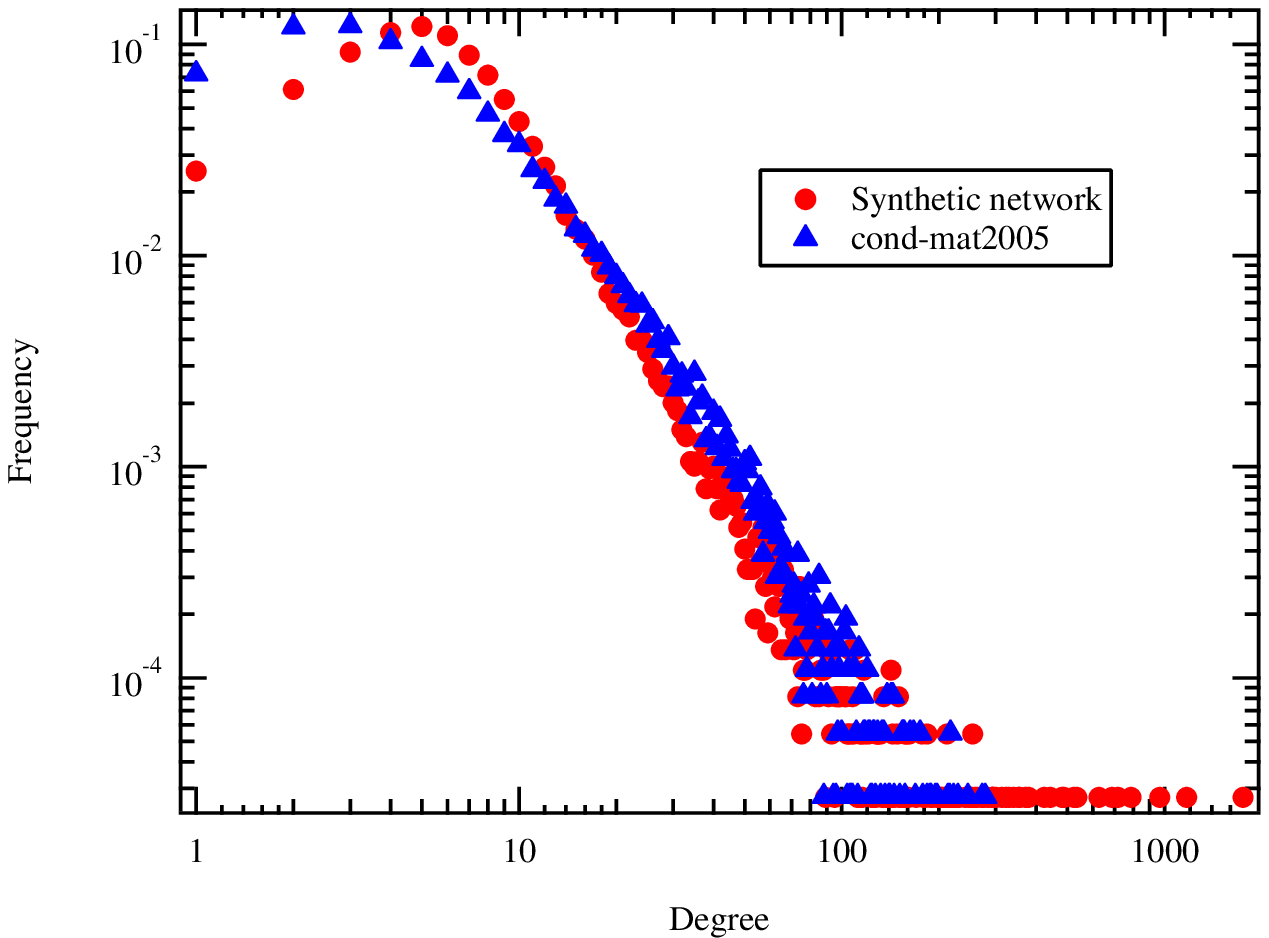}
  \caption{Cond-mat2005:: Collaboration network of scientists posting preprints on the condensed matter archive at arXiv, 1995-1999.}\label{cond-mat2005}
\end{figure}
\begin{figure}
  \centering
  \includegraphics[scale=0.5]{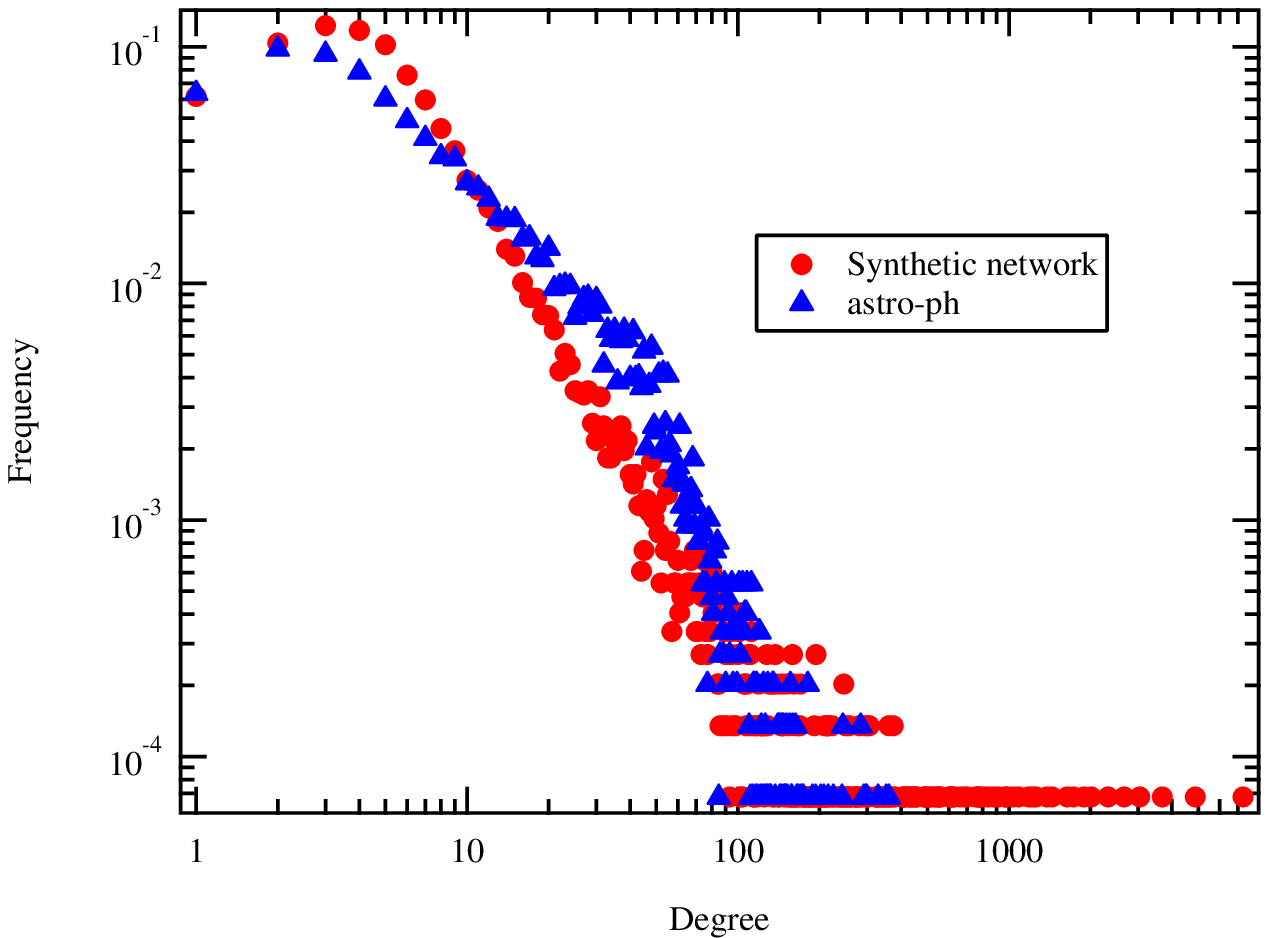}
  \caption{Astro-ph: Network of co-authorship between scientists posting preprints on the Astrophysics E-Print Archive between Jan 1, 1995 and December 31, 1999.}\label{astro-ph}
\end{figure}
\begin{figure}
  \centering
  \includegraphics[scale=0.5]{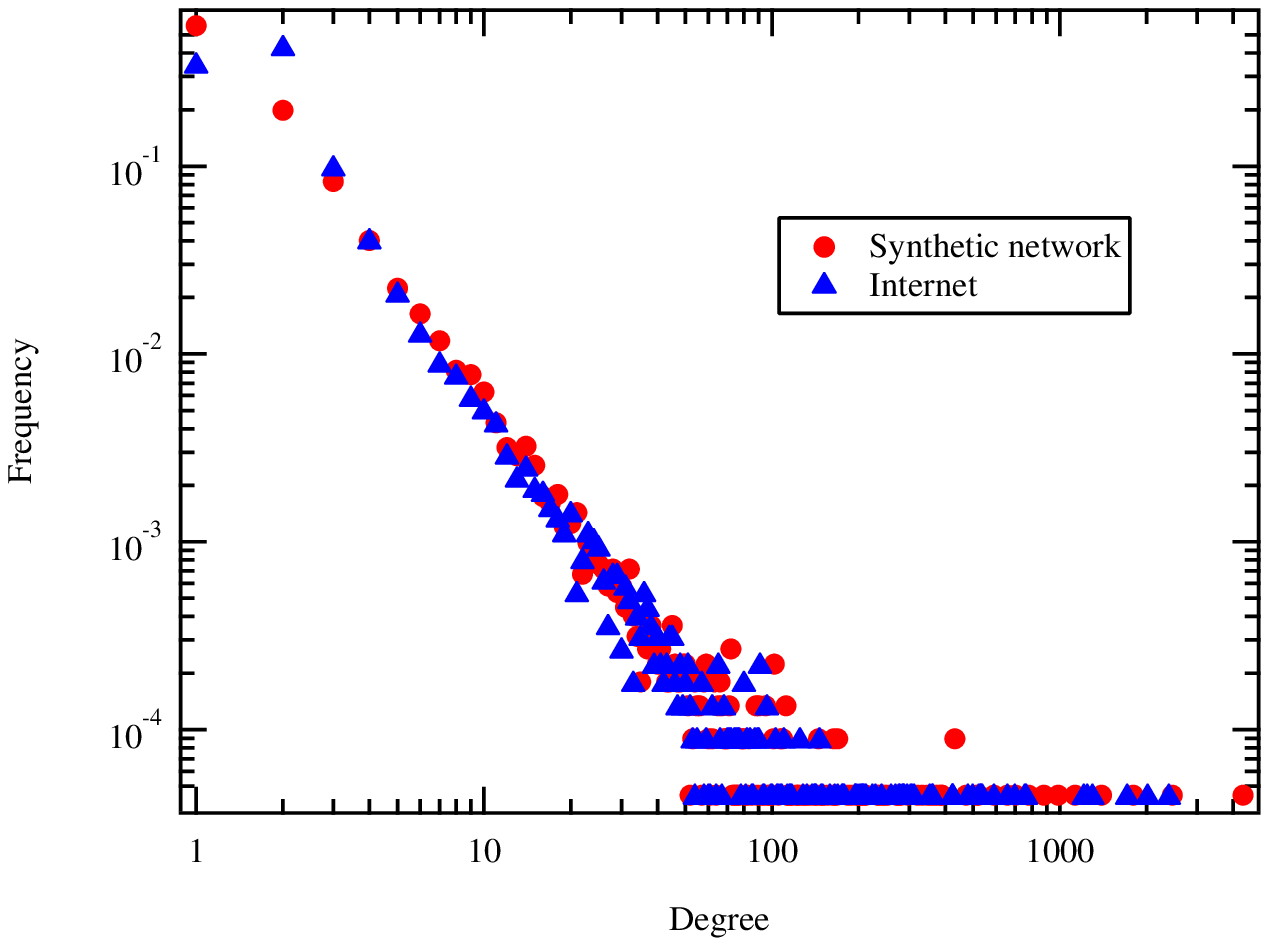}
  \caption{Internet: a symmetrized snapshot of the structure of the Internet at the level of autonomous systems, reconstructed from BGP tables posted by the University of Oregon Route Views Project.}\label{Internet}
\end{figure}
\begin{figure}
  \centering
  \includegraphics[scale=0.5]{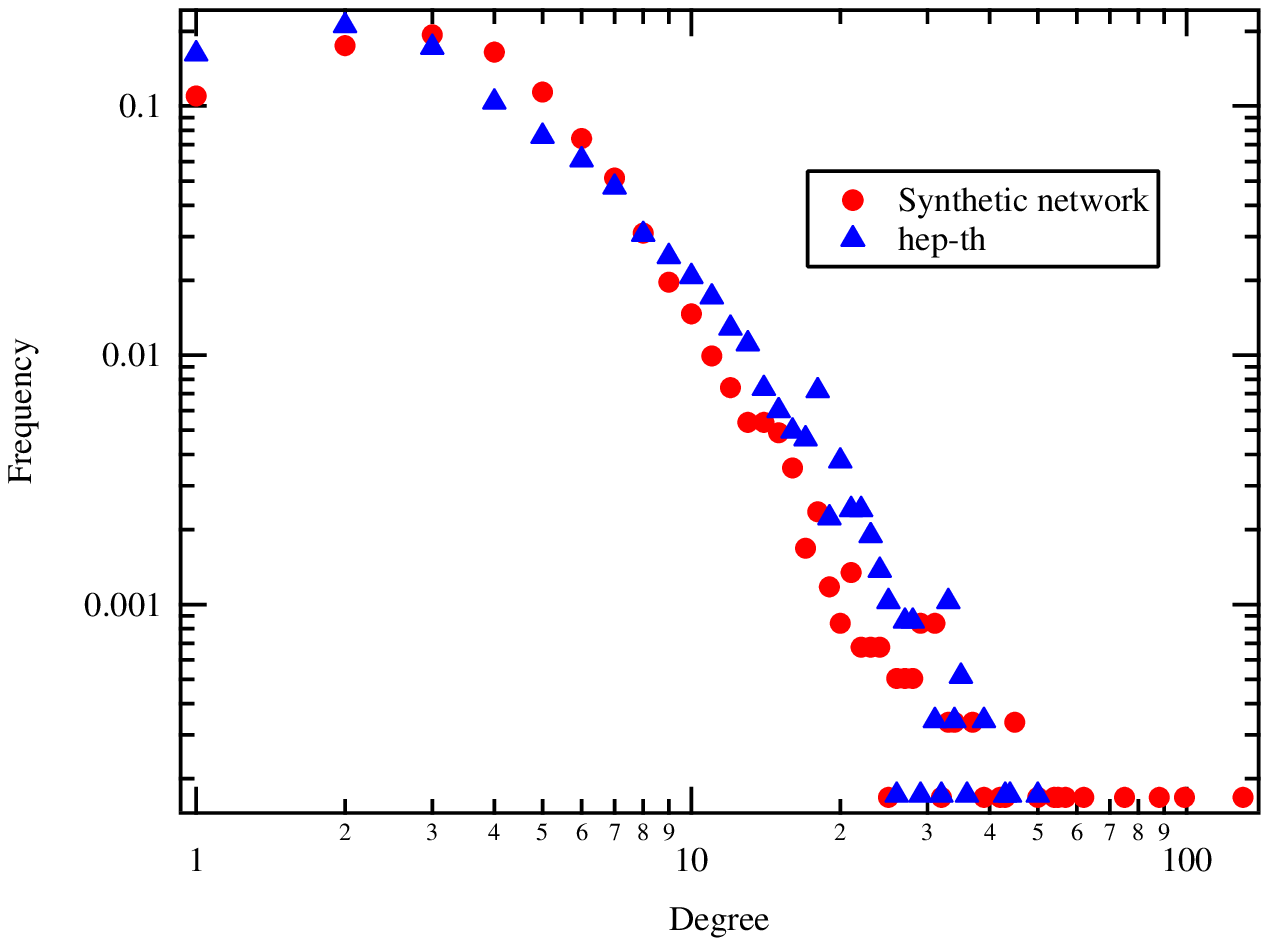}
  \caption{Hep-th: Network of co-authorship between scientists posting preprints on the High-Energy Theory arXiv between Jan 1, 1995 and December 31, 1999.}\label{hepth}
\end{figure}
\begin{figure}
  \centering
  \includegraphics[scale=0.5]{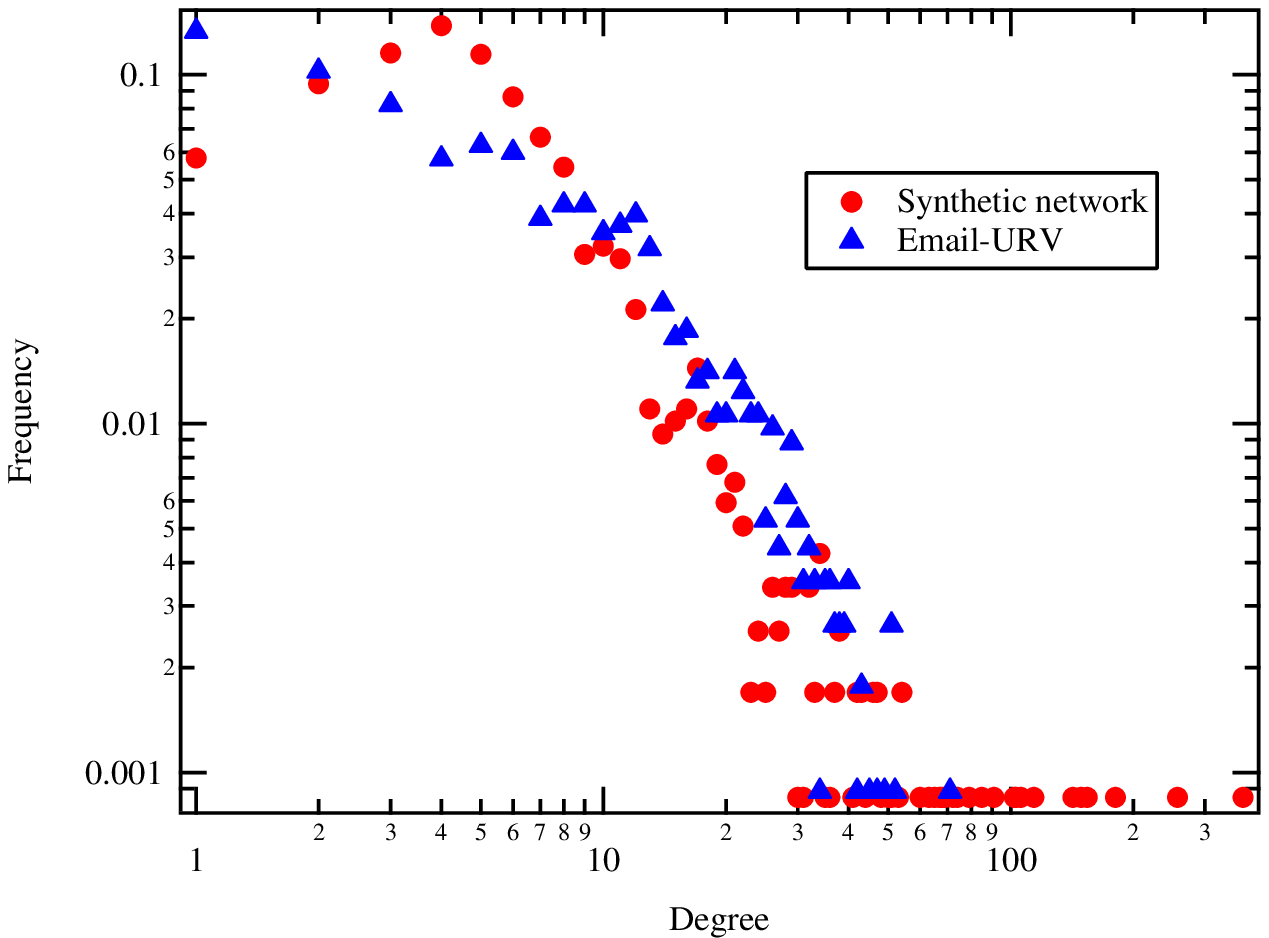}
  \caption{Email-URV: Network of E-mail interchanges between members of the Univeristy Rovira i Virgili, Tarragona.}\label{EmailURV}
\end{figure}
\begin{figure}
  \centering
  \includegraphics[scale=0.5]{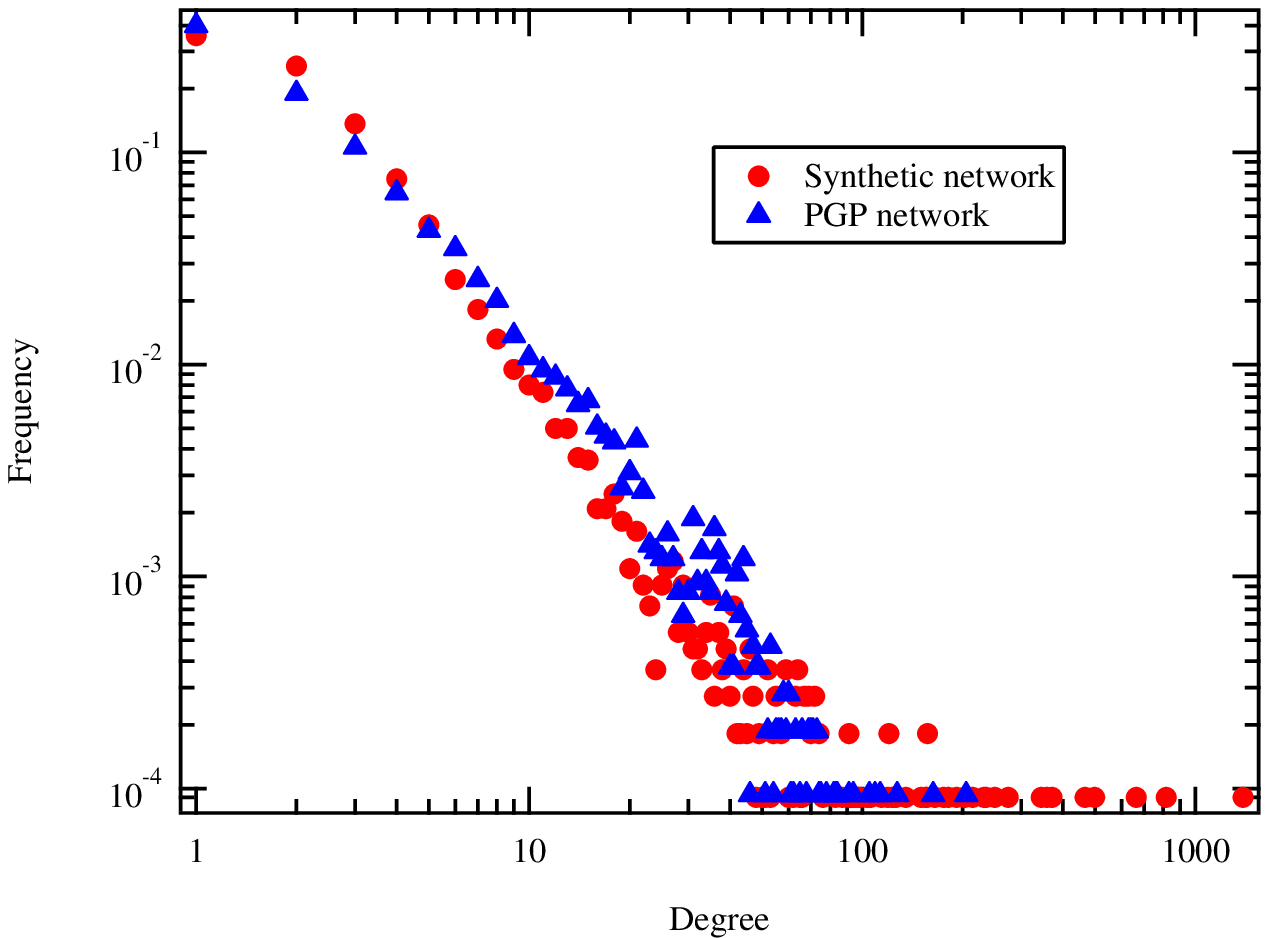}
  \caption{PGP: Network of users of the Pretty-Good-Privacy algorithm for secure information interchange.}\label{PGP}
\end{figure}
\begin{figure}
  \centering
  \includegraphics[scale=0.5]{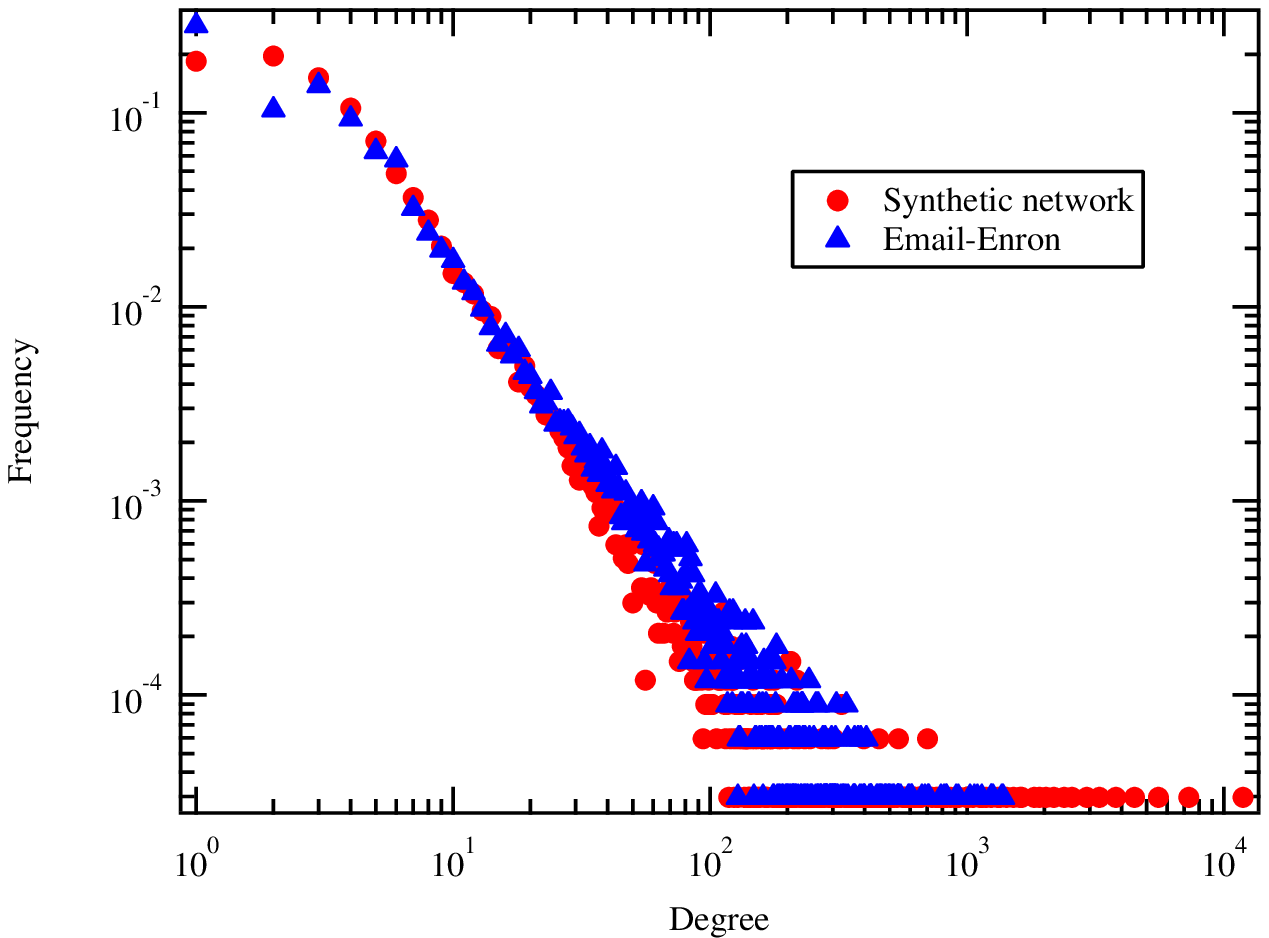}
  \caption{Email-Enron: Enron email communication network.}\label{EmailEnron}
\end{figure}

\bibliographystyle{IEEEtran}
\bibliography{bibitem}